\documentclass[conference]{IEEEtran}
\usepackage{epsfig,latexsym,amsmath,amssymb,setspace,cite,color,graphicx,algorithm,algorithmic,cases}
\usepackage{amssymb,amsmath,amsfonts,bm,epsfig,graphicx,latexsym}
\usepackage{rotating,setspace,latexsym,epsf,color}
\usepackage{cite,authblk,epstopdf,footnote}
\usepackage{float}
\usepackage{tabularx}
\usepackage{array}
\usepackage[subrefformat=parens,labelformat=parens]{subfig}
\usepackage{bbm}
\usepackage{amsthm}
\usepackage{thmtools}
\usepackage{algorithm,algorithmic}

\declaretheorem[style=plain,name=Definition,qed=$\blacksquare$]{Definition}

\declaretheorem[style=plain,name=Remark,qed=$\blacksquare$]{remark}

\declaretheorem[style=plain,qed=$\square$]{theorem}
\allowdisplaybreaks
\def\mc{\ensuremath\mathcal}

\begin{document}

%\sloppy

%% Paper Title
%% You can use linebreaks \\ within to get better formatting as
%% desired. 
\title{Benefits of Coded Placement for Networks with Heterogeneous Cache Sizes}
\author{Abdelrahman M. Ibrahim}
\author{Ahmed A. Zewail}
\author{Aylin Yener}
\affil{\normalsize Wireless Communications and Networking Laboratory (WCAN)\\
Electrical Engineering and Computer Science Department\\
The Pennsylvania State University, University Park, PA 16802.\\
\em \{ami137,zewail\}@psu.edu \qquad yener@engr.psu.edu}

\maketitle

\begin{abstract}
In this work, we study coded placement in caching systems where the users have unequal cache sizes and demonstrate its performance advantage. In particular, we propose a caching scheme with coded placement for three-user systems that outperforms the best caching scheme with uncoded placement. In our proposed scheme, users cache both uncoded and coded pieces of the files, and the coded pieces at the users with large memories are decoded using the unicast/multicast signals intended to serve users with smaller memories. Furthermore, we extend the proposed scheme to larger systems and show the reduction in delivery load with coded placement compared to uncoded placement.
\end{abstract}
%\footnotetext[1]{This work was done when Abdelrahman M.Ibrahim was with the Wireless Intelligent Networks Center (WINC), Nile University, Giza, Egypt.}
%\footnotetext[2]{This material is based upon work supported by the Marie Curie International Research Staff Exchange Scheme Fellowship PIRSES-GA-2010-269132 AGILENet within the 7th European Community Framework Program.}
%\vspace{-0.12in}
% ------------------------------------------------------------------------
% ------------------------------------------------------------------------
\section{Introduction}
Coded caching \cite{maddah2014fundamental} alleviates network congestion during peak-traffic hours, known as the \textit{delivery phase}, by placing some of the data in the cache memories at the network edge during off-peak hours, known as the \textit{placement phase}. Reference \cite{maddah2014fundamental} has shown that the joint design of the two phases leads to significant reduction in the delivery load, which is achieved by designing the cache contents in the placement phase in a manner that allows serving the users using multicast transmissions in the delivery phase. In turn, there exists a fundamental trade-off between the delivery load on the server and the cache sizes in the network, which has been studied in several setups \cite{maddah2014fundamental,yang2018coded,sengupta2016layered,
ibrahim2017centralized,ibrahim2018coded,ibrahim2017optimization,bidokhti2017benefits,
ibrahim2018distortion,tian2017uncoded,gomez2016fundamental}. Recently, references \cite{yang2018coded,sengupta2016layered,ibrahim2017centralized,ibrahim2018coded} have studied the effect of heterogeneity in cache sizes at the users on the delivery load memory trade-off. In particular, we have shown that placement and delivery schemes jointly optimized with respect to given cache sizes provide significant improvement over schemes tailored to uniform cache sizes. Additionally, cache sizes at the end users may be optimized for further gain \cite{ibrahim2017optimization,bidokhti2017benefits}. All of these consider placement of partial files in the caches, i.e., uncoded placement.

As eluded to above coded caching schemes are often categorized according to whether in the placement phase coding over the files is utilized or not. In caching with uncoded placement, the server places uncoded pieces of each file in the cache memories of the users \cite{maddah2014fundamental,sengupta2016layered,ibrahim2017centralized,ibrahim2018coded}. Alternatively, in systems with \textit{coded placement}, the server places coded pieces of the files in the users' caches which are decoded using the transmissions in the delivery phase\cite{tian2017uncoded,gomez2016fundamental}. While uncoded placement is sufficient for some systems, clearly, coding over files in general has the potential to perform better.

For systems with equal cache sizes, references \cite{tian2017uncoded,gomez2016fundamental} have shown that coded placement is beneficial in the small memory regime when the number of files is less than or equal the number of users. Recent reference \cite{cao2018coded} has shown that coded placement is essential in achieving the optimal delivery load in a two-user system when the cache sizes of the two users differ.

In this work, we show that coded placement in systems where the users are equipped with heterogeneous cache sizes, outperforms uncoded placement. We show that coded placement not only increases the local caching gain at the users, but also increases the multicast gain in the system. In particular, we propose a caching scheme with coded placement for three-user systems that illustrates the reduction in the worst-case delivery load compared to the best caching scheme with uncoded placement in \cite{ibrahim2018coded}. In our proposed scheme, users cache both uncoded and coded pieces of the files, and users with large memories decode the cached coded pieces using the transmissions intended to serve users with smaller memories. We observe that the gain from coded placement increases as the difference between the cache sizes increases, and decreases with the number of files. 
We extend the proposed scheme to systems with more than three end-users and show the reduction in delivery load with coded placement. This work thus establishes the first result showing a clear benefit of coded placement in a heterogeneous caching systems with three or more users.

\begin{figure}[t]
\includegraphics[scale=1.2]{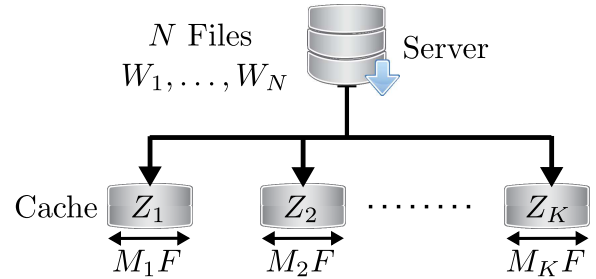}
\centering
\caption{Caching system with heterogeneous cache sizes.}\label{fig:sys_model}
\vspace{-0.2in}
\end{figure} 

\textit{Notation:} Vectors are represented by boldface letters, $ \oplus$ refers to the binary XOR operation, $|W|$ denotes cardinality of $W$, $\mc A \setminus \mc B $ denotes the set of elements in $\mc A$ and not in $\mc B $, $[K] \triangleq \{1,\dots,K\}$, and $ \phi$ denotes the empty set.

%--------------------------------------------------------------------------------------
%--------------------------------------------------------------------------------------

\section{System model}\label{sec_sysmod}
We consider a caching system where a single server is connected to $K$ users via a shared error-free multicast link \cite{maddah2014fundamental}, as shown in Fig. \ref{fig:sys_model}. The server has access to a library $\{ W_{1}, \dots, W_{N}\}$ of $N$ independent files, each with size $F$ symbols over the field $ \mathbb{F}_{2^{r}}$. We consider a heterogeneous system, where user $k$ is equipped with a cache memory of size $M_k F$ symbols. Without loss of generality, we assume that $M_1 \leq M_2 \leq \dots \leq M_K$.
Additionally, we define $m_k=M_k/N$ to denote the memory size of user $k$ normalized by the library size $N F$, i.e., $m_k \in [0,1]$ for $M_k \in [0,N]$.

The system operates over two phases: placement phase and delivery phase. 
In the placement phase, the server populates the users' cache memories without the knowledge of users' demands that will be made in the delivery phase. The cached content by user $k$ is denoted by $Z_k$, i.e., $|Z_k|\leq M_kF$. In the delivery phase, user $k$ requests a file $W_{d_k}$ from the server. The users' demands are uniform and independent as in \cite{maddah2014fundamental}. To serve the users' demands, the server transmits a sequence of unicast/multicast signals, $X_{\mc T, \bm d}$. At the end of the delivery phase, user $k$ must be able to decode $\hat W_{d_k}$ reliably. Formally, for a given normalized cache size vector $\bm m \triangleq [m_1,\dots,m_K]$, the worst-case delivery load $R(\bm m) \triangleq  \sum_{\mc T} |X_{\mc T, \bm d}|/F$ is said to be achievable if for every $\epsilon > 0$ and large enough $F$, there exists a caching scheme such that $\max_{\bm d, k \in [K]} Pr(\hat W_{d_k}\neq W_{d_k})\leq \epsilon$. % Moreover, the infimum over all achievable delivery loads is denoted by $R^*(\bm m)$.  %we have the following definition.% , i.e., for any $\epsilon >0$, we have $\max_{\bm d, k \in [K]} Pr(\hat W_{d_k}\neq W_{d_k})\leq \epsilon. $ 
%\begin{Definition} 
%For a given normalized cache size vector $\bm m$, the delivery load $R(\bm m)$ is said to be achievable if for every $\epsilon > 0$ and large enough $F$, there exists a caching scheme such that $\max_{\bm d, k \in [K]} Pr(\hat W_{d_k}\neq W_{d_k})\leq \epsilon$. Moreover, the infimum over all achievable delivery loads is denoted by $R^*(\bm m)$. 
%\end{Definition}
Our achievability scheme utilizes maximum distance separable (MDS) codes which are defined as follows. 
\begin{Definition} \cite{macwilliams1977theory}
An $(n,k)$ maximum distance separable (MDS) code is an erasure code that allows recovering  $k$ initial information symbols from any $k$ out of the $n$ coded symbols. In a systematic $(n,k)$-MDS code the first $k$ symbols in the output codeword is the information symbols. That is, we have
\begin{align}
[i_1,\dots,i_k] \bm G_{k \times n} &= [i_1,\dots,i_k] [ \bm I_{k \times k} \ \bm P_{k \times n-k}] \nonumber \\ &=[i_1,\dots,i_k,c_{k+1},\dots,c_n ],
\end{align}
where $\bm G_{k \times n}$ is the code generator matrix and $\bm I_{k \times k}$ is an identity matrix. 
\end{Definition}
For a systematic $(2N-j,N)$ MDS-code, we define 
\begin{align}
\sigma_j([i_1,\dots,i_N]) \triangleq  [i_1,\dots,i_N] \bm P_{N \times N-j}
\end{align}
to denote the $N-j$ parity symbols in the output codeword. Note that $\sigma_j([i_1,\dots,i_N]) $ represents $N-j$ independent equations in the information symbols $[i_1,\dots,i_N]$. For example, $\sigma_1([i_1,\dots,i_N]) = [i_1 \oplus i_2, i_2 \oplus i_3, \dots, i_{N-1} \oplus i_N]$.

% \cite{dutta2018optimal}

%--------------------------------------------------------------------------------------
%--------------------------------------------------------------------------------------

\section{Results}\label{sec_results}
In this section, we present our results showing the reduction in the delivery load when coded placement is utilized. In Theorem \ref{thm_3UE}, we characterize an achievable delivery load for three-user systems, that is lower than the minimum worst-case delivery load under uncoded placement, given by  
\begin{align}\label{eqn_3ue_uncoded}
R^*_{uncoded}&(\bm m) = \max \Big\lbrace  3 \! - \! 3 m_1 \! - \! 2 m_2 \! - \! m_3,  \nonumber \\ & \! \! \! \! \frac{5}{3} \! - \! \frac{3 m_1 \! - \! 2 m_2 \! - \! m_3}{3}\! ,   2  -  2 m_1 \! - \! m_2,  \  1  -  m_1  \Big\rbrace,
\end{align}
which has been characterized in \cite{ibrahim2018coded}.
\begin{theorem}\label{thm_3UE}
For a three-user system with $N \geq 4$ and $m_1 \leq m_2 \leq m_3 $, the worst-case delivery load   
\begin{align}\label{eqn_thm_3ue}
R_{\text{coded}}&(\bm m, N) =\! \max \Big\lbrace  3 \! - \! 3 m_1 \! - \! 2 m_2 \! - \! m_3 \! - \! \frac{3(m_2 \! - \! m_1)}{N-1} \! \nonumber \\ & - \! \frac{2(m_3 \! - \!m_2)}{N-2}, \ \frac{5}{3} \! - \! \frac{3 m_1 \! - \! 2 m_2 \! - \! m_3}{3} \! - \! \frac{m_2 \! - \! m_1}{3(N-1)}, \nonumber \\ &  \ 2  -  2 m_1 \! - \! m_2 \! - \! \frac{m_2 \! - \! m_1}{N \! - \! 1}, \  1  -  m_1  \Big\rbrace \! ,
\end{align}
is achievable with coded placement. 
\end{theorem}
\begin{proof}
The reduction in the delivery load in (\ref{eqn_thm_3ue}) compared to (\ref{eqn_3ue_uncoded}) is achieved by placing coded pieces of the files at users $2$ and $3$, which are decoded in the delivery phase. For example, in order to achieve $R_{\text{coded}}(\bm m, N)=R^*_{uncoded}(\bm m)\! - \! \frac{m_2 \! - \! m_1}{3(N-1)}, $ part of the multicast signal to users $\{1,2\}$ is utilized in decoding the cached pieces at user $3$. The proposed caching scheme is presented in Section \ref{sec_cach_3UE}.
\end{proof}
Next theorem characterizes the gain achieved by coded placement in the small memory regime, where the unicast signals intended for users $\{1, \dots,k\}$ are utilized in decoding the cache content at users $\{k+1,..,K\}$.
\begin{theorem}\label{thm_small_mem}
For a $K$-user system with $N \geq K+1$, $m_1 \leq m_2 \leq \dots \leq m_K $, and $$ \sum_{i=1}^{K} m_i + \sum_{i=2}^{K} \frac{(i-1)(K-i+1)(m_i-m_{i-1})}{N-i+1} \leq 1,$$ the worst-case delivery load   
\begin{align}
R_{\text{coded}}&(\bm m, N, K) = R^*_{uncoded}(\bm m) \! \nonumber \\ & - \! \sum_{i=2}^{K} \frac{(i \! - \! 1)(K \! - \! i \! + \! 1)(K \! - \! i \! + \! 2)(m_i -m_{i-1})}{2(N-i+1)}, 
\end{align}
is achievable with coded placement, where 
\begin{align}
R^*_{uncoded}(\bm m) =K \! - \! \sum_{i=1}^{K} (K-i+1) m_i,
\end{align}
is the minimum worst-case delivery load with uncoded placement for $\sum_{i=1}^{K} m_i \leq 1$ \cite{ibrahim2017centralized,ibrahim2018coded}.
\end{theorem}
\begin{proof}
The caching scheme is detailed in Section \ref{sec_cach_small_mem}.
%In Section \ref{sec_cach_small_mem}, we show that the unicast signals to users $\{1,\dots,k\} $ can be used in decoding the cached coded pieces at users $\{k\!+\!1,\dots,K\}$
\end{proof}
\begin{remark} For given $K$ and $\bm m$, $\lim\limits_{N \rightarrow \infty} R_{\text{coded}}(\bm m, N,K) = R^*_{uncoded}(\bm m,K)$. That is, the gain due to coded placement decreases with $N$ and is negligible for $N > > K$. 
\end{remark}

%--------------------------------------------------------------------------------------
%--------------------------------------------------------------------------------------
\section{Proof of Theorem \ref{thm_3UE}}\label{sec_cach_3UE}
In this section, we present our caching schemes for there-user systems. The achievable delivery load in Theorem \ref{thm_3UE} consists of the following regions
\begin{itemize}
\item Region $\rm I$: If $\sum_{i=1}^3 m_i \! + \! \frac{2(m_2\! - \!m_1)}{N-1} \! + \! \frac{2(m_3\! - \!m_2)}{N-2} \leq 1$, then 
$ R_{\text{coded}} = 3 \! - \! 3 m_1 \! - \! 2 m_2 \! - \! m_3 \! - \! \frac{3(m_2 \! - \! m_1)}{N-1} \! - \! \frac{2(m_3 \! - \!m_2)}{N-2}. $
\item Region $\rm II$: If $\sum_{i=1}^3 m_i  +  \frac{2(m_2 - m_1)}{N-1}  +  \frac{2(m_3 - m_2)}{N-2} > 1$, $N m_3  \leq  (N  +  3) m_2  +  3(N  -  2) m_1  -  (N -  1),$ and $N m_3  \leq  2(N  -  1)-(2 N  -  3) m_2,$ then $ R_{\text{coded}} = 5/3  -  m_1 - 2 m_2/3  -  m_3/3  -  (m_2 \! - \! m_1)/(3(N-1)). $
\item Region $\rm III$: If $\sum_{i=1}^3 m_i  +  \frac{2(m_2 - m_1)}{N-1}  +  \frac{2(m_3 - m_2)}{N-2} > 1$, $N m_3  >  (N  +  3) m_2  +  3(N  -  2) m_1  -  (N -  1),$ and $N m_2  +(N-2) m_1 \leq N-1,$ then $ R_{\text{coded}} = 2  -  2m_1 - m_2 -  (m_2 \! - \! m_1)/(N-1). $
\item Region $\rm IV$: If $N m_2  +(N-2) m_1 > N-1,$ and $N m_3  >  2(N  -  1)-(2 N  -  3) m_2,$ then $ R_{\text{coded}} = 1  -  m_1. $
\end{itemize}
Region $\rm I$ is a special case of Theorem \ref{thm_small_mem} which will be explained in Section \ref{sec_cach_small_mem}. Next, we consider regions $\rm II$ to $\rm IV$. %The caching scheme in region $\rm I$ is presented in section \ref{sec_cach_small_mem}.  
%%%%%%%%%%%%%%%%%%%%%%%%%%%%%%%%%%%%%%%%%%%%%%%%%%%%%%%%%%%%%%%%%%%%%%%%%%%%%%%%%%%%%%%
\subsection{Region $\rm II$}

\subsubsection{Placement Phase} Each file $W_n$ is split into subfiles $W_{n,1},$ $W_{n,\{2\}},$ $W_{n,\{3\}},$ $W_{n,\{1,2\}},$ $\big\{ W_{n,\{1,3\}}^{(1)},W_{n,\{1,3\}}^{(2)} \big\},$ and $\big \{ W_{n,\{2,3 \}}^{(1)}, W_{n,\{2,3 \}}^{(2)}, W_{n,\{2,3 \}}^{(3)}, W_{n,\{2,3 \}}^{(4)} \big\},$ such that 
\begin{align}
|W_{n,\{1\}}|&=\Big( \frac{2}{3}-m_1-\frac{N(m_3 \! - \! m_2)}{3(N \! - \! 1)} \Big) F, \\
|W_{n,\{2\}}|&=|W_{n,\{3\}}|=|W_{n,\{1\}}|-(m_2-m_1)F, \\
|W_{n,\{2,3 \}}^{(3)}|&=|W_{n,\{2,3 \}}^{(4)}|=(m_2-m_1)F, \\
|W_{n,\{1,2\}}|&=|W_{n,\{1,3 \}}^{(1)}|=|W_{n,\{2,3 \}}^{(1)}| \nonumber \\  &=\Big( m_1-1/3-\frac{N(m_3 \! - \! m_2)}{3(N \! - \! 1)} \Big) F, \\
|W_{n,\{1,3 \}}^{(2)}|&=|W_{n,\{2,3 \}}^{(2)}|=\frac{N(m_3 \! - \! m_2)}{(N \! - \! 1)} F,
\end{align}
where all subfiles are cached uncoded except for $ W_{n,\{2,3 \}}^{(2)}, \forall n$ are encoded before being placed at user $3$. More specifically, the cache contents are given as 
\begin{align}
Z_1 \! &= \! \bigcup_{n=1}^{N} \! \Big( W_{n,\{1\}} \! \bigcup \! W_{n,\{1,2\}} \! \bigcup \! W_{n,\{1,3\}}^{(1)} \! \bigcup \! W_{n,\{1,3\}}^{(2)} \Big), \\
Z_2 \! &= \! \bigcup_{n=1}^{N} \! \Big( W_{n,\{2\}} \! \bigcup \! W_{n,\{1,2\}} \! \bigcup \! \Big( \bigcup_{i=1}^{4} W_{n,\{2,3\}}^{(i)}  \Big) \Big), \\
Z_3 \! &= \! \bigcup_{n=1}^{N} \! \Big( W_{n,\{3\}} \! \bigcup \!  \Big( \bigcup_{i=1}^{2} W_{n,\{1,3\}}^{(i)}  \Big) \! \bigcup \! \Big( \bigcup_{i=1,i \neq 2}^{4} W_{n,\{2,3\}}^{(i)}  \Big) \Big) \nonumber \\ &   \ \ \ \ \bigcup   \sigma_1 \big( [W_{1,\{2,3 \}}^{(2)}, \dots,W_{N,\{2,3 \}}^{(2)} ]\big),.
\end{align}
%where the encoding function $\gamma$ maps $2N$ information symbols to $2N-1$ encoded symbols and is defined as
%\begin{align}
%\! \! \gamma \Big( \! [ i_1,i_2,\dots,i_{2N}] \! \Big) \! \triangleq \! [ i_1 \! \oplus \! i_2,i_2 \! \oplus \! i_3,\dots,i_{2N-1} \! \oplus \! i_{2N}],
%\end{align}
%where the addition $\oplus$ is over the finite field $\mathbb{F}_q$.
%
%\gamma \Big( \! \big[ \! W_{1,\{1,3 \}}^{(2)},\dots, W_{N,\{1,3 \}}^{(2)}, W_{1,\{2,3 \}}^{(2)}, \dots, W_{N,\{2,3 \}}^{(2)} \! \big] \! \Big)
%\begin{align}
%\! \! \gamma \Big( \! [ i_1,i_2,\dots,i_{2N}] \! \Big) \! \triangleq \! [ i_1,i_2,\dots,i_{2N}] \! \begin{bmatrix}
%1, 0, 0, \dots, 0,0\\
%1, 1, 0, \dots, 0,0\\
%0, 1, 1, \dots, 0,0\\
%\vdots \\
%0, 0, 0, \dots, 1,1\\
%0, 0, 0, \dots, 0,1
%\end{bmatrix}
%\end{align}

\subsubsection{Delivery Phase} The server sends the following multicast signals
\begin{align}
X_{\{1,2\} , \bm d}^{\prime}&=  W_{d_2,\{1\}} \oplus \Big( W_{d_1,\{2\}} \bigcup W_{d_1,\{2,3\}}^{(3)} \Big), \\
X_{\{1,3\} , \bm d}&=  W_{d_3,\{1\}} \oplus \Big( W_{d_1,\{3\}} \bigcup W_{d_1,\{2,3\}}^{(4)} \Big), \\
X_{\{2,3\} , \bm d}&=  W_{d_3,\{2\}} \oplus W_{d_2,\{3\}}, \\
X_{\{1,2,3\} , \bm d}&=  W_{d_3,\{1,2\}} \oplus W_{d_2,\{1,3\}}^{(1)} \oplus W_{d_1,\{2,3\}}^{(1)}, \\
X_{\{1,2\} , \bm d}^{\prime \prime}&=  W_{d_2,\{1,3\}}^{(2)} \oplus W_{d_1,\{2,3\}}^{(2)}.
\end{align}

\subsubsection{Achievability} The proposed placement scheme is valid since the cache sizes constraints are satisfied. In the delivery phase, the users retrieve the requested pieces from the multicast signals using the cached subfiles. Additionally, using the multicast signal $X_{\{1,2\} , \bm d}^{\prime \prime}$ and the cached piece $W_{d_2,\{1,3\}}^{(2)}$, user $3$ decodes $W_{d_1,\{2,3\}}^{(2)}$, which is used in retrieving $W_{d_3,\{2,3\}}^{(2)}$ from $\sigma_1 \big( [W_{1,\{2,3 \}}^{(2)}, \dots,W_{N,\{2,3 \}}^{(2)} ]\big) $.

%In order for user $3$ to retrieve $W_{d_3,\{2,3\}}^{(2)}$ from its cache, we go through the following steps.
%\begin{itemize}
%\item User $3$ constructs the equation $W_{d_2,\{1,3\}}^{(2)} \oplus W_{d_1,\{2,3\}}^{(2)}$ from the encoded pieces in its cache.
%\item Using the multicast signal $X_{\{1,2\} , \bm d}^{\prime \prime}$ and $W_{d_2,\{1,3\}}^{(2)} \oplus W_{d_1,\{2,3\}}^{(2)}$, user $3$ decodes $W_{d_1,\{2,3\}}^{(2)} $ and $W_{d_2,\{1,3\}}^{(2)} $.
%\item User $3$ retrieves $W_{d_3,\{1,3\}}^{(2)}$ and $W_{d_3,\{2,3\}}^{(2)}$ from its cache using $W_{d_2,\{1,3\}}^{(2)} $ and $W_{d_1,\{2,3\}}^{(2)} $.  
%\end{itemize}
%For instance, if we consider the binary extension field $\mathbb{F}_{2^2}=\{0,1,\beta,1+\beta\} $, user $3$ constructs $W_{d_2,\{1,3\}}^{(2)} \oplus W_{d_1,\{2,3\}}^{(2)}$ by adding a subset of the encoded pieces in its cache. Additionally, choosing $\beta_1=1$ and $\beta_2=\beta$ enables user $3$ to solve for $W_{d_2,\{1,3\}}^{(2)} $ and $ W_{d_1,\{2,3\}}^{(2)}$. 
%%%%%%%%%%%%%%%%%%%%%%%%%%%%%%%%%%%%%%%%%%%%%%%%%%%%%%%%%%%%%%%%%%%%%%%%%%%%%%%%%%%%%%%
\subsection{Region $\rm{III}$}
\subsubsection{Placement Phase} Each file $W_n$ is split into $W_{n,1},$ $\{ W_{n,\{2\}}^{(1)}, W_{n,\{2\}}^{(2)} \}, \{ W_{n,\{3\}}^{(1)}, W_{n,\{3\}}^{(2)}, W_{n,\{3\}}^{(3)} \},$ $W_{n,\{1,3\}},$ and $\big\{ W_{n,\{2,3\}}^{(1)},W_{n,\{2,3\}}^{(2)} \big\},$ such that 
\begin{align}
&|W_{n,\{1\}}| \! + \! |W_{n,\{1,3\}}|=|W_{n,\{2\}}^{(1)}| \! + \! |W_{n,\{2,3\}}^{(1)}|=m_1 F, \\
&|W_{n,\{1,3\}}|=|W_{n,\{2,3\}}^{(1)}|=m_1 F-|W_{n,\{3\}}^{(1)}|, \\
&|W_{n,\{2\}}^{(2)}|=|W_{n,\{3\}}^{(2)}|=\frac{N(m_2-m_1)}{N-1} F-|W_{n,\{2,3\}}^{(2)}|. 
\end{align}
In particular, we have the following three cases.
\begin{itemize}
\item For $m_3 \leq \! \frac{N \! - \! 2}{N} \! + \! \frac{m_2}{N-1} \! - \! \frac{(2N \! - \! 3)(N \! - \! 2)m_1}{(N-1)(N)}$
\begin{align}
& \! \! \! \! \! \! \! \! \! \! \! \! \! \! |W_{n,\{2,3\}}^{(2)}| \! = \! \! \Big(  \! \sum_{i=1}^{3} \! \! m_i \! + \! \frac{2(m_2 \! - \! m_1)}{N-1} \! + \! \frac{2(m_3 \! - \! m_2)}{N-2} \! - \! 1 \! \Big) \! F, \! \\
& \! \! \! \! \! \! \! \! \! \! \! \! \! \! |W_{n,\{3\}}^{(3)}|=\frac{N(m_3-m_2)}{N-2}F, \  |W_{n,\{3\}}^{(1)}|=m_1 F. 
\end{align}
\item For $ \frac{m_2}{N-1} \! - \! \frac{(2N \! - \! 3)(N \! - \! 2)m_1}{(N-1)(N)} < m_3 \! - \!  \frac{N \! - \! 2}{N} \! \leq  \! \frac{m_2}{N-1} \! - \! \frac{m_1}{(N-1)(N)} $
\begin{align}
\! \! \! \! \! \! |W_{n,\{2,3\}}^{(2)}| \! &= \frac{N(m_2 \! - \! m_1)}{N-1} F,\\
\! \! \! \! \! \! |W_{n,\{3\}}^{(3)}| \! &= \! \Big( 1 \! - \! 2 m_1 \! - \! \frac{N(m_2 \! - \! m_1)}{N-1} \Big)F \! - \! |W_{n,\{3\}}^{(1)}|, \\
\! \! \! \! \! \! |W_{n,\{3\}}^{(1)}|&= \Big( \frac{N \! - \! 2}{2N \! - \! 3} \Big) \Big( 1 \! - \! \frac{(N \! - \! 3) m_1}{(N-2)} \! - \! \frac{N(m_2 \! - \! m_1)}{N-1} \nonumber \\ 
& \ \ \ \ \ - \! \frac{N(m_3 \! - \! m_2)}{N-2} \Big) F.
\end{align}
\item For $ m_3 \! >  \! \frac{N \! - \! 2}{N} \! + \! \frac{m_2}{N-1} \! - \! \frac{m_1}{(N-1)(N)} $
\begin{align}
|W_{n,\{3\}}^{(3)}| \! &= \Big( 1 \! - \! 2 m_1 \! - \! \frac{N(m_2 \! - \! m_1)}{N-1} \Big)F, \\
|W_{n,\{3\}}^{(1)}| \! &= \! 0, \ |W_{n,\{2,3\}}^{(2)}| \! = \! \frac{N(m_2 \! - \! m_1)}{N-1} F,
\end{align}
\end{itemize}
The cache contents are defined as
\begin{align}
Z_1 \! &= \! \bigcup_{n=1}^{N} \! \Big( W_{n,\{1\}} \bigcup W_{n,\{1,3\}} \Big), \\
Z_2 \! &= \! \bigcup_{n=1}^{N} \! \Big( W_{n,\{2\}}^{(1)} \bigcup W_{n,\{2,3\}}^{(1)} \Big) \bigcup \sigma_1 \big( [W_{1,\{2\}}^{(2)}, \dots,W_{N,\{2 \}}^{(2)} ]\big) \nonumber \\ & \ \ \ \bigcup \sigma_1 \big( [W_{1,\{2,3 \}}^{(2)}, \dots ,W_{N,\{2,3 \}}^{(2)} ]\big), \\
Z_3 \! &= \! \bigcup_{n=1}^{N} \! \Big( W_{n,\{3\}}^{(1)} \bigcup W_{n,\{2,3\}}^{(1)} \Big) \bigcup \sigma_1 \big( [W_{1,\{3\}}^{(2)}, \dots,W_{N,\{3 \}}^{(2)} ]\big) \nonumber \\ & \ \ \ \bigcup \sigma_1 \big( [W_{1,\{2,3\}}^{(2)}, \dots ,W_{N,\{2,3 \}}^{(2)} ]\big) \bigcup \sigma_1 \big( [W_{1,\{1,3\}}, \dots \nonumber \\ &  \ \ ,W_{N,\{1,3 \}}]\big) \bigcup \sigma_2 \big( [W_{1,\{3\}}^{(3)}, \dots ,W_{N,\{3 \}}^{(3)} ]\big).
\end{align}

\subsubsection{Delivery Phase:} The server sends the following multicast signals
\begin{align}
\! \! \! \! X_{\{1,2\} , \bm d}& \! =  \! \! \Big( \! W_{d_2,\{1\}} \! \bigcup \! W_{d_2,\{1,3\}} \! \Big) \! \! \oplus \! \! \Big( \! W_{d_1,\{2\}}^{(1)} \! \bigcup \! W_{d_1,\{2,3\}}^{(1)} \! \Big), \! \! \\
\! \! \! \! X_{\{2,3\} , \bm d}& \! =  \! \! \Big( \! W_{d_3,\{2\}}^{(1)} \! \bigcup \! W_{d_3,\{2\}}^{(2)} \! \Big) \!  \oplus  \! \Big( \! W_{d_2,\{3\}}^{(1)} \! \bigcup \! W_{d_2,\{3\}}^{(2)} \! \Big), \! \! \\
\! \! \! \! X_{\{1,3\} , \bm d}& \! = W_{d_3,\{1\}} \oplus W_{d_1,\{3\}}^{(1)}. 
\end{align}
The following unicast signals complete the requested files and help users $\{2,3\}$ in decoding their cache contents.
\begin{align}
X_{\{1\} , \bm d}&  =    W_{d_1,\{2\}}^{(2)}  \bigcup  W_{d_1,\{2,3\}}^{(2)}  \bigcup  W_{d_1,\{3\}}^{(2)} \bigcup W_{d_1,\{3\}}^{(3)}, \\
X_{\{2\} , \bm d}& = W_{d_2,\{3\}}^{(3)}.
\end{align}

\subsubsection{Achievability} User $2$ decodes its cache using $W_{d_1,\{2\}}^{(2)}, W_{d_1,\{2,3\}}^{(2)}$ from $X_{\{1\} , \bm d}$. Similarly, user $3$ decodes its cache using $W_{d_1,\{2,3\}}^{(2)},  W_{d_1,\{3\}}^{(2)},  W_{d_1,\{3\}}^{(3)} $ from $X_{\{1\} , \bm d}$ and  $W_{d_2,\{3\}}^{(3)} $ from $X_{\{2\} , \bm d} $.

%%%%%%%%%%%%%%%%%%%%%%%%%%%%%%%%%%%%%%%%%%%%%%%%%%%%%%%%%%%%%%%%%%%%%%%%%%%%%%%%%%% 
\subsection{Region $\rm{IV}$} 
Next, we consider the case where $ m_1 + m_2 \leq 1$, since uncoded placement is optimal for $ m_1 + m_2 > 1$ \cite{ibrahim2018coded}.

\subsubsection{Placement Phase} Each file $W_n$ is split into $W_{n,\{1,2\}},$ $W_{n,\{1,3\}},$ and $\big\{ W_{n,\{2,3\}}^{(1)},W_{n,\{2,3\}}^{(2)},W_{n,\{2,3\}}^{(3)} \big\},$ such that 
\begin{align}
|W_{n,\{1,2\}}|&= m_1 F - |W_{n,\{1,3\}}|, \\
|W_{n,\{1,2\}}|&=|W_{n,\{2,3\}}^{(2)}|, \ |W_{n,\{1,3\}}|=|W_{n,\{2,3\}}^{(1)}|, \\
|W_{n,\{2,3\}}^{(3)}|&=(1-2m_1)F. 
\end{align}
In particular, we have
\begin{align}
\! \! \! \! |W_{n,\{1,2\}}| \! = \! \begin{cases} \! \!  \Big( \! \frac{N m_2  +  (N \! - \! 2) m_1}{N \! - \! 1} \! - \! 1\Big) F, \ \! \text{if} \ \! m_3 \! \leq \! \frac{(N \!- \!1)  +  m_1}{N}, \! \! \! \! \! \! \\
0, \text{ otherwise.}
\end{cases}
\end{align}
The cache contents are defined as
\begin{align}
Z_1 \! &= \! \bigcup_{n=1}^{N} \! \Big( W_{n,\{1,2\}} \bigcup W_{n,\{1,3\}} \Big), \\
Z_2 \! &= \! \bigcup_{n=1}^{N} \! \Big( W_{n,\{2,3\}}^{(1)} \bigcup W_{n,\{2,3\}}^{(2)} \Big) \bigcup \sigma_1 \big( [W_{1,\{1,2\}}, \dots, \nonumber \\ & \ \ \ W_{N,\{1,2 \}} ]\big)  \bigcup \sigma_1 \big( [W_{1,\{2,3 \}}^{(3)}, \dots ,W_{N,\{2,3 \}}^{(3)} ]\big),\\
Z_3 \! &= \! \bigcup_{n=1}^{N} \! \Big( W_{n,\{2,3\}}^{(1)} \bigcup W_{n,\{2,3\}}^{(2)} \Big) \bigcup \sigma_1 \big( [W_{1,\{1,3\}}, \dots, \nonumber \\ & \ \ \ W_{N,\{1,3 \}} ]\big)  \bigcup \sigma_1 \big( [W_{1,\{2,3 \}}^{(3)}, \dots ,W_{N,\{2,3 \}}^{(3)} ]\big).
\end{align}
\subsubsection{Delivery Phase} The server sends the following signals
\begin{align}
X_{\{1,2\} , \bm d} &= W_{d_2,\{1,3\}} \oplus W_{d_1,\{2,3\}}^{(1)}, \\
X_{\{1,3\} , \bm d} &= W_{d_3,\{1,2\}} \oplus W_{d_1,\{2,3\}}^{(2)}, \\
X_{\{1\} , \bm d} &= W_{d_1,\{2,3\}}^{(3)}.
\end{align}
\subsubsection{Achievability} User $2$ retrieves $W_{d_2,\{1,2\}} $ from its cache using $W_{d_3,\{1,2\}}$ which is extracted from $X_{\{1,3\} , \bm d}$. Similarly, user $3$ retrieves $W_{d_3,\{1,3\}} $ by utilizing $X_{\{1,2\} , \bm d}$.
%--------------------------------------------------------------------------------------
%--------------------------------------------------------------------------------------
\section{Proof of Theorem \ref{thm_small_mem}}\label{sec_cach_small_mem}
In this section, we explain the caching schemes that achieves the delivery load in Theorem \ref{thm_small_mem}.
\subsection{Placement Phase}
File $W_n$ is divided into $K(K \! + \! 1)/2  +  1 $ subfiles, $W_{n,\phi},W_{n,1}^{(1)},$ $\big\{ W_{n,2}^{(1)},W_{n,2}^{(2)} \big\},$ $\dots,$ $\big\{W_{n,K}^{(1)},\dots ,W_{n,K}^{(K)}\big\}$, such that  
\begin{align}
&|W_{n,k}^{(1)}|=m_1 F, \ \forall n, \\
&|W_{n,k}^{(i)}|= \frac{N(m_{i} \! - \! m_{i-1})}{(N\! - \! i \! + \! 1)}  F, \ i=2,\dots,k, \forall n,  \\
 &|W_{n,\phi} | \! = \! \Big( \! 1 \! - \! \! \! \sum_{i=1}^{K} \! \! m_i \! - \! \! \! \sum_{i=2}^{K} \! \! \frac{(i\! - \! 1)(K \! - \! i \! + \! 1)(m_i \! - \! m_{i-1})}{N-i+1} \Big) \! F, \! \! 
\end{align}
% &|W_{n,\phi} | \! = \! \Big( \! 1 \! - \! \! \! \sum_{i=1}^{K} \! \! m_i \! - \! \! \! \sum_{i=2}^{K} \! \! \frac{(i\! - \! 1)(K \! - \! i \! + \! 1)(m_i \! - \! m_{i-1})}{N-i+1} \Big) \! F, \! \! 
where $W_{n,\phi} $ is available only at the server. User $k$ caches subfiles $W_{1,k}^{(1)}, \dots, W_{N,k}^{(1)}$ uncoded and the MDS encoded pieces $\sigma_{i-1} \big([W_{1,k}^{(i)},\dots,W_{N,k}^{(i)} ]\big)$ for $i=2,\dots,k$. In turn, the cache contents of user $k$ is defined as   
\begin{align}
Z_k= \Big( \bigcup_{n=1}^{N} W_{n,k}^{(1)} \Big) \bigcup  \Big( \bigcup_{i=2}^{k}  \sigma_{i-1} \big([W_{1,k}^{(i)},\dots,W_{N,k}^{(i)}]\big) \Big),
\end{align}
%which is illustrated in Fig. for $K=3$.  

\subsection{Delivery Phase}

The server sends the following unicast signals
\begin{align}
X_{\{K\} , \bm d}&=W_{d_K,\phi}, \\
X_{\{k\} , \bm d}&= \bigcup_{i=k+1}^{K} \bigcup_{j=i}^{K}  W_{d_{k},j}^{(i)} , \ \forall k\in [K \!- \!1], 
\end{align}
where the unicast signals $X_{\{1\} , \bm d},\dots,X_{\{k\} , \bm d} $ are used by users $\{k\!+\!1,\dots, K \}$ in decoding their cache contents.

Next, the server sends the pairwise multicast signals
\begin{align}
X_{\{k,j\} , \bm d}= \Big( \bigcup_{i=1}^{k} W_{d_j,k}^{(i)} \Big) \oplus \Big( \bigcup_{i=1}^{k} W_{d_k,j}^{(i)} \Big),
\end{align}
for $k=1,\dots,K$ and $j=k+1,\dots,K$. In turn, the delivery load is given as
\begin{align}
R_{\text{coded}}&= \sum_{k=1}^{K} |X_{\{k\} , \bm d}|/F + \sum_{k=1}^{K} \sum_{j=k+1}^{K} |X_{\{k,j\} , \bm d}|/F, \\
&= K \! - \! \sum_{i=1}^{K} (K-i+1) m_i \! \nonumber \\ & - \! \sum_{i=2}^{K} \frac{(i \! - \! 1)(K \! - \! i \! + \! 1)(K \! - \! i \! + \! 2)(m_i -m_{i-1})}{2(N-i+1)}.
\end{align}

\subsection{Achievability}
First, the cache size constraints are satisfied, since  
\begin{align}
|Z_k|&=N m_1 F + \sum_{i=2}^{k} (N \! - \! i \! + \! 1) |W_{n,k}^{(i)}|, \\
&= M_1 F + \sum_{i=2}^{k} (M_i -M_{i-1}) F= M_k F.
\end{align}
In the delivery phase, user $k$ reconstructs the requested file $W_{d_k} $ by going through the following steps.
\begin{itemize}
\item Subfile $W_{d_k,k}^{(1)} $ is uncoded and in turn can be directly retrieved from the cache memory.
\item Subfiles $W_{d_1,k}^{(i)},\dots,W_{d_{i-1},k}^{(i)} $ are extracted from the unicast signals $X_{\{1\} , \bm d},\dots,X_{\{k-1\} , \bm d}$. 
\item Subfile $W_{d_k,k}^{(i)}$ is retrieved from the encoded pieces $ \sigma_{i-1} \big([W_{1,k}^{(i)},\dots,W_{N,k}^{(i)}]\big) $ using $W_{d_1,k}^{(i)},\dots,W_{d_{i-1},k}^{(i)} $. 
\item Subfiles $W_{d_k,j}^{(k+1)},\dots,W_{d_k,j}^{(j)} $ where $j \in \{k+1,\dots,K\}$ and $W_{d_k, \phi} $ are retrieved from the unicast signal $X_{\{k\} , \bm d}$. 
\item Subfiles $W_{d_k,j}^{(1)},\dots,W_{d_k,j}^{(k)} $ are retrieved from the multicast signals $X_{\{k,j\} , \bm d}, j \in [K]\setminus \{k\} $. 
\end{itemize}

%--------------------------------------------------------------------------------------
%--------------------------------------------------------------------------------------
%\vspace{+0.1in}
\section{Numerical Results}

In Fig. \ref{fig:comp_dlv}, we compare the worst-case delivery load achieved by exploiting coded placement with the minimum worst-case delivery load assuming uncoded placement in a three-user system where $N=4$ and $m_k = \alpha \ \! m_{k+1}$. Fig. \ref{fig:comp_dlv} shows that the gain achieved by coded placement increases with the heterogeneity in cache sizes. We also observe that the gain is higher when the total memory is small.

\begin{figure}[t]
%\vspace{-.1 in}
\includegraphics[scale=0.475]{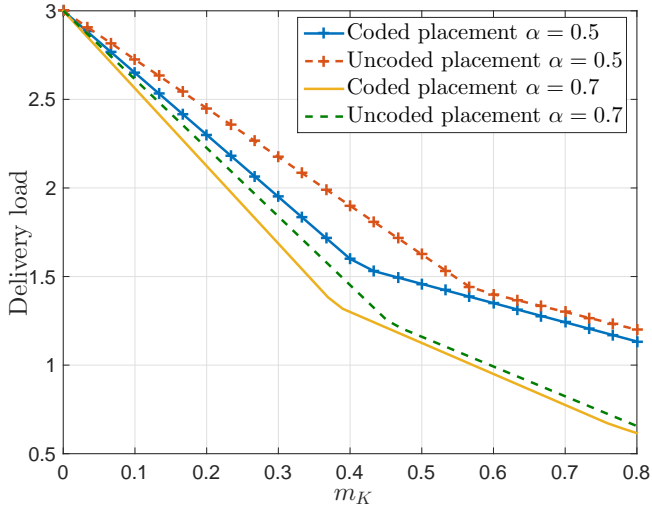}
\centering
\caption{Comparing the achievable delivery load assuming coded placement with the minimum delivery load under uncoded placement, for $K=3$, $N=4$, and $ m_k= \alpha\ \! m_{k+1}$.}\label{fig:comp_dlv}
\end{figure}

The delivery load achieved by utilizing coded placement in Theorem \ref{thm_small_mem} is compared to the best uncoded placement scheme in Fig. \ref{fig:comp_dlv2}, for $K=10$ and $m_k= 0.7 m_{k+1}$. From Fig. \ref{fig:comp_dlv2}, we observe that the reduction in the delivery load due to coded placement decreases with the number of files $N$. In turn, for a system where $N > > K$, the delivery load achieved with our coded placement scheme is approximately equal to the minimum delivery load under uncoded placement. That is, the coded placement gain is negligible when $N > > K $.

\begin{figure}[t]
%\vspace{-.1 in}
\includegraphics[scale=0.475]{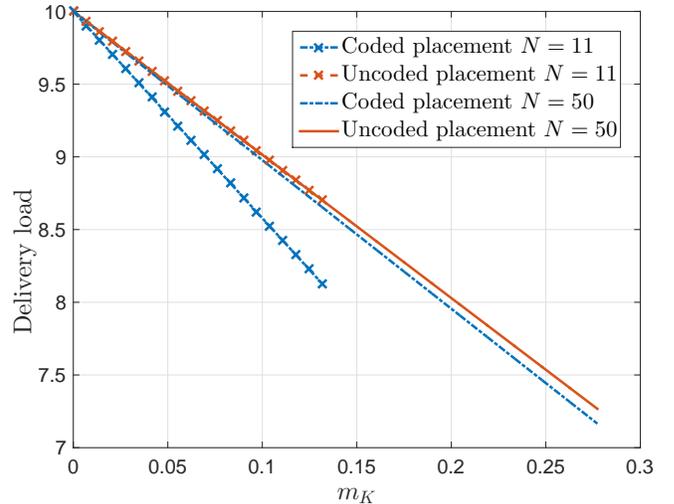}
\centering
\caption{Comparing the achievable delivery load assuming coded placement with the minimum delivery load under uncoded placement, for $K=10$, $\alpha=0.7$, and $ m_k= \alpha\ \! m_{k+1}$.}\label{fig:comp_dlv2}
\end{figure}

\section{Conclusion}\label{sec_con}
%\vspace{-0.02in}
In this paper, we have shown that coded placement leads to significant reduction in the delivery load in systems where the users have different cache sizes. In particular, we have proposed novel coded placement schemes that outperform the best uncoded placement schemes for three-user systems with arbitrary cache sizes and $K$-user systems where the cache sizes satisfy
$ \sum\limits_{i=1}^{K} m_i + \sum\limits_{i=2}^{K} (i \! - \! 1)(K \! - \! i \! + \! 1)(m_i \! - \! m_{i-1})/(N \! - \! i \! + \! 1) \leq 1.$\vspace{+0.05in} 

Our proposed schemes illustrate that the signals intended to serve users with small cache sizes can be used in decoding the cache contents of users with larger cache sizes. Furthermore, we have shown that the gain due to coded placement increases with the heterogeneity in cache sizes and decreases with the number of files.  

%
%This work is a step towards developing caching schemes for systems with practical considerations, such as wireless systems where the users have distinct download rates and non-uniform resources. Future directions include optimizing over cache sizes and consideration of incomplete/inaccurate demand profiles.
%
%\vspace{-0.1in}

\vspace{+0.04in}
\bibliographystyle{IEEEtran}
\bibliography{IEEEabrv,DiversityLib}
%\vspace{-0.05in}

\end{document}